\documentclass[10pt,journal,twocolumn]{IEEEtran}

\usepackage{graphicx}
\usepackage{amssymb}
\usepackage{amsmath}
\usepackage{bm}
\usepackage{dsfont}

\usepackage{float}
\usepackage{mathrsfs}
\usepackage{subfigure}
\usepackage{times}
\usepackage{xspace}

\newtheorem{prop}{Proposition}

\newenvironment{mat}[1]{\left[\begin{array}{#1}}{\end{array}\right]}
\newcommand{\bmx}[1]{\begin{mat}{#1}}
\newcommand{\emx}{\end{mat}}

\newcommand{\gss}[3]{\mbox{\boldmath $#1$}_{#2}^{#3}}
\newcommand{\g}[1]{\mbox{\boldmath $#1$}}

\newcommand{\sbm}[1]{\mbox{\scriptsize $\g{#1}$}}
\newcommand{\sbms}[2]{\mbox{\scriptsize $\gss{#1}{#2}{}$}}

\newcommand{\pref}[1]{(\ref{#1})}

\begin{document}
%
\title{On the Design of Channel Estimators for given Signal Estimators and Detectors}
%
\author{Dimitrios Katselis, Cristian R. Rojas, H\aa kan Hjalmarsson, Mats
Bengtsson, Mikael Skoglund
\thanks{The authors are with ACCESS Linnaeus Center, School of
Electrical Engineering, KTH Royal Institute of Technology, SE
100-44, Stockholm, Sweden. E-mail: {dimitrik@kth.se},
{cristian.rojas@ee.kth.se}, {hjalmars@kth.se},
{mats.bengtsson@ee.kth.se}, {mikael.skoglund@ee.kth.se}.}}

\maketitle

\begin{abstract}
The fundamental task of a digital receiver is to decide the
transmitted symbols in the best possible way, i.e., with respect
to an appropriately defined performance metric. Examples of usual
performance metrics are the probability of error and the Mean
Square Error (MSE) of a symbol estimator. In a coherent receiver,
the symbol decisions are made based on the use of a channel
estimate. This paper focuses on examining the optimality of usual
estimators such as the minimum variance unbiased (MVU) and the minimum mean square error (MMSE) estimators for these metrics and on proposing better
estimators whenever it is necessary. For illustration purposes, this study is performed on a toy channel model, namely a single input single output (SISO) flat fading
 channel with additive white Gaussian noise (AWGN). In this way, this paper highlights
the design dependencies of channel estimators on target performance metrics.
\end{abstract}


\section{Introduction}
\label{sec:intro} \PARstart{S}{ignal} estimation and detection are
two main concerns in the course of designing a communication
system~\cite{c07,g05,pro95}. The main goal is to design optimal
demodulators at the receiver side providing the detector with the
necessary sufficient statistics for its decision on the
transmitted symbol at a specific observation interval.
Furthermore, the optimization of the decision device is also a
target, i.e., its design based on such statistical tests which
rely on sufficient statistics and minimize the probability of
error. A different setup of optimal designs related to radar and
sonar systems is to detect the presence of either a deterministic
or random signal in noise with least probability of error or false
alarm~\cite{k98}. Although the two aforementioned setups have
conceptual differences, they are usually treated in the same
fashion. First, an optimal demodulator is necessary to deliver the
sufficient statistics to the decision device. Then, the decision
device, that optimally uses these sufficient statistics, has to be
derived. The optimal design of the decision device is formulated
in any case as a hypotheses testing problem. Moreover, the
optimization of the transmitter is another related problem. In
this case, the problem turns to be the design of optimal
transmission sets, such that the end performance metric, i.e, the
probability of error is minimized.

Depending on the degree of knowledge about the transmission
channel at the receiver side, the detector can be coherent,
semi-coherent or noncoherent~\cite{pro95}. The more information
about the transmission channel is available, the better the receiver's performance will be. This
justifies the fact that the receivers usually have a built-in
channel estimator. In the communication and signal processing
literature, the usual channel estimators are the minimum variance
unbiased (MVU) and the minimum mean square error (MMSE)
estimators \cite{k93}.  The combination of these channel
estimators with the optimal decision devices is usually considered
to address the problem of determining the optimal receiver.

Current physical layer (PHY) standards that have attracted a lot of attention both from
the mobile industry and the research community
are the Wireless Interoperability for Microwave
Access (WiMAX), the Long Term Evolution (LTE) and the Digital Video Broadcasting (DVB) either in
its terrestrial (DVB-T) or its Handheld (DVB-H) versions \cite{agm07,flypz12,gak12,hb12,pcz12,s11,sy12}. These standards
are orthogonal frequency division multiple access (OFDMA) based and they can satisfy the
need for shorter communication links to provide truly broadband connectivity services. In these systems,
either MVU/least squares (LS) or MMSE channel estimators are used, usually
employing some sort of estimate interpolation through the frame if the goal is to
track a time-varying channel \cite{ase10,chcl03,hxy11,lcc12,lhp09,ls12,smhl02,sacbf10,yryq12}.

In this paper, we re-examine the validity of the common belief
that the MVU and MMSE channel estimators are the best choices to
be combined with the optimal detectors, delivering an overall
optimal receiver, when finite-sample training is used to estimate
the channel\footnote{In this sense, the asymptotic efficiency of
the maximum likelihood (ML) estimator together with its invariance property are
irrelevant.}. To this end, ideas originating from the system
identification field are employed. Recent results in optimal
experiment design indicate that it is better to design the
optimal training for the estimation of a certain set of unknown
parameters with respect to optimizing the end performance metric
rather than the mean square error of the parameter estimator
itself~\cite{bsghh06,h09,jh05,krhb12}. We will slightly modify
this idea and we will examine if the aforementioned channel
estimators are the best choices, when the selection of the channel
estimator is made with respect to an appropriately defined end
performance metric. For illustration purposes, this study is performed on a toy channel model, namely a single input single output (SISO) flat fading
 channel with additive white Gaussian noise (AWGN)\footnote{In this toy model, the MVU estimator coincides with the LS and the ML channel estimators.}. The initial focus is on two
different MSE criteria. These MSE criteria serve to demonstrate
the dependence of the optimal channel estimators on the end
performance metrics. Their choice is based on the simplicity of
the analysis that they allow. Then, using the obtained results, we
will examine the case of the error probability as
the performance metric of interest. We show that for several
performance metrics examined in this paper, the MVU and MMSE
channel estimators are suboptimal, while we propose ways to obtain
better channel estimators. Finally, we numerically compare the
performances of the derived channel estimators with those of the
MVU and MMSE channel estimators for all performance metrics in
this paper. These comparisons verify that the optimality of the
usual channel estimators with respect to common end performance
metrics is questionable.

This paper is organized as follows: Section~\ref{sec:ProbSt}
defines the problem of designing the channel estimator with
respect to the end performance metric. Section~\ref{sec:prelim}
presents some results and comments that will be useful in the rest
of the paper, while it introduces approximations of the
performance metrics that the rest of the analysis will be based
on. The optimality of the MVU and MMSE channel estimators with
respect to the minimization of the symbol estimate MSE is examined in Section~\ref{sec:dMSE} and subsections
therein, while uniformly better channel estimators are also proposed. The
same analysis as in Section~\ref{sec:dMSE} is pursued in
Section~\ref{sec:eMSE} for a differently defined symbol estimate MSE and in
Section~\ref{sec:minPe} for a rough approximation (variation) of the error probability
performance metric. Section~\ref{sec:sims} illustrates the
validity of the derived results. Finally, Section~\ref{sec:concl} concludes
the paper.

\section{Problem Statement}
\label{sec:ProbSt}

The received signal model for a SISO
system, when the channel is considered to be narrowband block
fading, is given as follows:
\begin{equation}
y(n)=hx(n)+w(n), \label{eq:sm}
\end{equation}
where $y(n)$ is the observed signal at the receiver side at time
instant $n$, $h$ is the complex channel impulse response
coefficient, $x(n)$ is the transmitted symbol at the same time
instant taken from an M-ary constellation
$\mathcal{X}=\{x_1,x_2,\ldots, x_M\}$ and $w(n)$ is
complex, circularly symmetric, Gaussian noise with zero mean and
variance $\sigma_w^2$. Given an equiprobable distribution on the
constellation symbols, we further assume that $E[x(n)]=0$ and
$E[|x(n)|^2]=\sigma_x^2$, while our modulation method is
memoryless. In addition, $w(n)$ and $x(n)$ are independent random
sequences, while $w(n)$ is a white random sequence.

Assume that a maximum energy $\mathcal{E}$ and a training length
of $B$ time slots are available at the transmitter for training.
We can collect the received samples corresponding to training in
one vector:
\begin{equation}
\gss{y}{\rm tr}{}=h\gss{x}{\rm tr}{}+\gss{w}{\rm tr}{},
\label{eq:rxTr}
\end{equation}
where $\gss{y}{\rm tr}{}=\left[y(l-B+1), y(l-B+2),\cdots, y(l)
\right]^{T}$ is the vector of $B$ received samples corresponding
to training, $\gss{x}{\rm tr}{}=\left[ x(l-B+1), x(l-B+2), \cdots,
x(l) \right]^{T}$ is the vector of $B$ training symbols and
$\gss{w}{\rm tr}{}=\left[w(l-B+1), w(l-B+2), \cdots, w(l) \right]^{T}$
is the vector of $B$ noise samples. Considering the class of
linear channel estimators, the channel is estimated as follows:
\begin{equation}
\hat{h}=\gss{f}{}{H}\gss{y}{\rm tr}{}=h\gss{f}{}{H}\gss{x}{\rm
tr}{}+\gss{f}{}{H}\gss{w}{\rm tr}{}, \label{eq:chEst}
\end{equation}
where $\g{f}$ is a $B\times 1$ channel estimating filter.

With the assumptions in \pref{eq:sm}, if the constellation symbols
are equiprobable and the channel is perfectly known, the ML
detector is optimal~\cite{c07,pro95}. This is with respect to
minimizing the probability that a different symbol from the one
transmitted is decided given the transmitted symbol. The ML
decision rule is given by the following expression:
\begin{equation}
{\rm dec}\left[x(n)\right](h)=\arg\min_{\hat{x}(n)\in
\mathcal{X}}|y(n)-h\hat{x}(n)|^2. \label{eq:MLdec}
\end{equation}
Here, ${\rm dec}\left[x(n)\right]$ denotes the decision of the
detector, when the transmitted symbol is $x(n)$. In essence, the
ML detector minimizes the probability of error, when the
transmitted symbols are equiprobable. When the receiver has a
channel estimate $\hat{h}$, $h$ is replaced by $\hat{h}$ in the
last expression.

A different kind of performance metric is the MSE of a
\emph{linear} symbol estimator. In this paper, we will call the
symbol estimator an \emph{equalizer}. The equalizer uses the
channel knowledge and delivers a soft decision of the transmitted
symbol, i.e., a symbol estimate. We will call \emph{clairvoyant}
the equalizer that has perfect channel knowledge. Denoting this
equalizer by $\tilde{c}(h)$, we can find its mathematical
expression as follows:
\begin{equation}
\tilde{c}(h)=\arg\min_{c(h)}
E\left[\left|c(h)y(n)-x(n)\right|^2\right], \label{eq:MMSEchoice}
\end{equation}
where the expectation is taken over the statistics of $x(n)$ and
$w(n)$. If we set the derivative of the last expression with
respect to $c(h)$ to zero and we solve for $c(h)$, then the
optimal clairvoyant equalizer is given by the expression
\begin{equation}
\tilde{c}(h)=\frac{\sigma_x^2 h^{*}}{|h|^2\sigma_x^2+\sigma_w^2}.
\label{eq:MMSEeq}
\end{equation}
We will call this the MMSE clairvoyant equalizer. We observe that
as the SNR increases, i.e., $\sigma_{w}^2\rightarrow 0$,
$\tilde{c}(h)\rightarrow 1/h$. We will call $\check{c}(h)=1/h$ the \emph{Zero Forcing}
(ZF) clairvoyant equalizer. Using the above definitions and
assuming that the receiver has only an estimate of the channel,
the system performance metric is the symbol estimate MSE:
\begin{equation}
{\rm
MSE}_{x}=E\left[\left|c(\hat{h})y(n)-x(n)\right|^2\right].\label{eq:MSEdirect}
\end{equation}

The MSE given by \pref{eq:MSEdirect} can be defined in two
different ways: If we assume that the channel is an unknown but
otherwise deterministic quantity, then the expectation in
\pref{eq:MSEdirect} does not consider $h$. This leads to an MSE
expression dependent on the unknown channel $h$. In this case,
only the channel estimators that treat the channel as an unknown
deterministic variable are meaningful. If we assume that the
unknown channel is a random variable, then we can average the MSE
expression over $h$. In this case, both the estimators that treat
the channel as an unknown deterministic variable or as a random
variable are meaningful. The former represents the case where the
system designer chooses to ignore the knowledge of the channel
statistics in the selection of the channel estimator for some
reason.

In the following, we focus on the ZF equalizer, which becomes
optimal as the SNR increases. This choice is made to preserve the
simplicity of this paper and to highlight the derived results.

The previous MSE definition implies the definition of yet another
MSE that is meaningful in the context of communication
systems. Given an equalizer, we can define the excess of the
symbol estimate based on an equalizer that only knows a channel
estimate over the equalizer with perfect channel knowledge, thus leading to
\begin{equation}
{\rm
MSE}_{xe}=E\left[\left|c(\hat{h})y(n)-c(h)y(n)\right|^2\right].\label{eq:MSEexcess}
\end{equation}
In the sequel, this metric will be called \emph{excess} MSE.

 Our goal will be to determine the optimal channel estimators
for fixed training sequences so that each performance metric based
on a given equalizer is minimized. To this end, the following section
presents some useful ideas.


\section{Preliminary Results}
\label{sec:prelim}

Consider the MVU estimator. Since it is an unbiased estimator,
it satisfies $\gss{f}{}{H}\gss{x}{\rm tr}{}=1$. This condition
implies that $E[\hat{h}]=h$. For our problem assumptions, the MVU
estimator can be found by solving the following optimization
problem:
\begin{eqnarray}
&& \min_{\sbm{f}} \sigma_w^2\|\g{f}\|^2\nonumber\\
&& {\rm s.t.}\ \ \gss{f}{}{H}\gss{x}{\rm tr}{}=1.
\end{eqnarray}
Forming the Lagrangian for this problem and zeroing its gradient
with respect to $\gss{f}{}{}$, we get:
\begin{equation}
\gss{f}{\rm MVU}{}=\frac{\gss{x}{\rm tr}{}}{\|\gss{x}{\rm
tr}{}\|^2}.\label{eq:MVUE}
\end{equation}
For the sake of completeness, this estimator coincides with the ML and LS channel estimators
under our assumptions.

If we assume that the prior distribution of $h$ is known, then
instead of the MVU one could use the MMSE channel estimator. With
our assumptions and the extra assumption that $E[h]=0$, one can
obtain~\cite{k93}
\begin{equation}
\gss{f}{\rm MMSE}{}=\frac{E[|h|^2]\gss{x}{\rm
tr}{}}{E[|h|^2]\|\gss{x}{\rm tr}{}\|^2+\sigma_w^2}.\label{eq:MMSE}
\end{equation}


The ${\rm MSE}_x$ of the ZF equalizer using a deterministic channel (``dc'') assumption is
\begin{equation}
{\rm MSE}_x^{dc}\left({\rm
ZF}\right)=E\left[\left|\frac{\hat{h}-h}{\hat{h}}\right|^2\right]\sigma_x^2+\sigma_w^2E\left[\frac{1}{|\hat{h}|^2}\right],\label{eq:MSEddcZF}
\end{equation}
the corresponding for random channel (``rc'') is:
\begin{equation}
{\rm MSE}_x^{rc}\left({\rm
ZF}\right)=E_h\left[E\left[\left|\frac{\hat{h}-h}{\hat{h}}\right|^2\right]\right]\sigma_x^2+\sigma_w^2E_h\left[E\left[\frac{1}{|\hat{h}|^2}\right]\right],\label{eq:MSEdrcZF}
\end{equation}
while for the ${\rm MSE}_{xe}$ we accordingly have:
\begin{equation}
{\rm MSE}_{xe}^{dc}\left({\rm
ZF}\right)=E\left[\left|\frac{\hat{h}-h}{\hat{h}}\right|^2\right]\left(\sigma_x^2+\frac{\sigma_w^2}{|h|^2}\right)\label{eq:eMSEZF}
\end{equation}
(c.f. \pref{eq:MSEexcess}). The ${\rm MSE}_{xe}^{rc}$ is
obtained by averaging the last expression over $h$.

Depending on the probability distributions of $|\hat{h}|$ and $|h|$,
the above MSE expressions may fail to exist. The
MSEs will be finite if the probability distribution function (pdf)
of $|\hat{h}|$ is of order $O(|\hat{h}|^2)$ as $\hat{h}\rightarrow
0$. A similar condition should hold for the pdf of $|h|$ in the
case of ${\rm MSE}_{xe}^{rc}$. In the opposite case, we end up with an \emph{infinite moment} problem. In order to obtain well-behaved channel estimators that will
be used in conjunction with the actual performance metrics, some
sort of regularization is needed. Some ideas for appropriate
regularization techniques to use may be obtained by modifying
robust estimators (against heavy-tailed distributions), e.g., by
trimming a standard estimator, if it gives a value very close to
zero \cite{hu05}. An example of such a trimmed estimator is given as follows:
\begin{eqnarray}
\hat{h}= \left\{%
\begin{array}{c}
  \gss{f}{}{H}\gss{y}{\rm tr}{},\ \ {\rm if}\ \ |\gss{f}{}{H}\gss{y}{\rm tr}{}|>\lambda\\
  \lambda
\gss{f}{}{H}\gss{y}{\rm tr}{}/|\gss{f}{}{H}\gss{y}{\rm tr}{}|,\ \ {\rm o.w.} \\
\end{array}
\right.\label{eq:WellBehEst}
\end{eqnarray}
where $\g{f}$ can be any estimator and $\lambda$ a regularization
parameter\footnote{This parameter can be tuned via
cross-validation or any other technique, although in the
simulation section we empirically select it for simplicity
purposes.}.

\emph{Remark:} Clearly, the reader may observe that the definition
of the trimmed $\hat{h}$ preserves the continuity at
$|\gss{f}{}{H}\gss{y}{\rm tr}{}|=\lambda$. Additionally, the event
$\{\gss{f}{}{H}\gss{y}{\rm tr}{}=0\}$ has zero probability since
the distribution of $\gss{f}{}{H}\gss{y}{\rm tr}{}$ is continuous.
Therefore, in this case $\hat{h}$ can be arbitrarily defined,
e.g., $\hat{h}=\lambda$.

We focus now on the ${\rm MSE}_{x}^{dc}({\rm ZF})$. Assume a fixed $\lambda$.
In the appendix, we show that,
for a sufficiently small $\lambda$ and a sufficiently high SNR during training,
minimizing ${\rm MSE}_x^{dc}({\rm ZF})$ is equivalent to
minimizing the following approximation
\begin{equation}
\left[{\rm MSE}_x^{dc}\left({\rm
ZF}\right)\right]_0=\frac{E\left[|\hat{h}-h|^2\right]}{E\left[|\hat{h}|^2\right]}\sigma_x^2+\sigma_w^2\frac{1}{E\left[|\hat{h}|^2\right]}.\label{eq:MSEddcZF0}
\end{equation}
Following similar steps and using some minor additional
technicalities, we can work with
\begin{equation}
\left[{\rm MSE}_x^{rc}\left({\rm
ZF}\right)\right]_0=\frac{E_h\left[E\left[|\hat{h}-h|^2\right]\right]}{E_h\left[E\left[|\hat{h}|^2\right]\right]}\sigma_x^2+\sigma_w^2\frac{1}{E_h\left[E\left[|\hat{h}|^2\right]\right]},\label{eq:MSEdrcZF0}
\end{equation}
instead of ${\rm MSE}_x^{rc}\left({\rm ZF}\right)$. Moreover,
$\left[{\rm MSE}_{xe}^{dc}\left({\rm ZF}\right)\right]_0$ and
$\left[{\rm MSE}_{xe}^{rc}\left({\rm ZF}\right)\right]_0$ can be
defined accordingly. We will call the last approximations
\emph{zeroth order} symbol estimate MSEs and excess MSEs,
respectively. The following analysis and results will be based on
the zeroth order  metrics and they will reveal the dependency
of the channel estimator's selection on the considered (any) end
performance metric.

\emph{Remarks:}
\begin{enumerate}
\item A useful, alternative way to consider the zeroth order MSEs is to view
them as affine versions of normalized channel MSEs, where the
actual true channel is $\hat{h}$ and the estimator is $h$.
\item In the definition of (\ref{eq:MSEddcZF0}), one can observe that after approximating the mean value of the ratio by the ratio of the mean values the infinite moment problem is eliminated. In the following, all zeroth order metrics will be defined based on the \emph{non-trimmed} $\hat{h}$ to ease the derivations. This treatment is approximately valid when $\lambda$ is sufficiently small as it is actually shown in eq. (\ref{eq:MSExdcApp}) of the appendix.
\end{enumerate}

\section{Minimizing the zeroth order Symbol Estimate MSE}
\label{sec:dMSE}

We now examine the zeroth order symbol estimate MSE in the case of
the ZF equalizer. The optimality of the MVU and MMSE channel
estimators will be investigated. Additionally, the training sequence
is assumed fixed.

\subsection{ZF Equalization}
\label{subsec:dMSEdcZF}

The channel is considered either deterministic or random,
depending on the available knowledge of a priori channel
statistics and the will of the system designer to ignore or to
exploit this knowledge.

\subsubsection{Deterministic Channel}
\label{subsec:dMSEdcZF-1}

The expectation operators in Eq. \pref{eq:MSEddcZF0} are with
respect to $\gss{w}{\rm tr}{},x(n)$ and $w(n)$. We have:
\begin{eqnarray}
\left[{\rm MSE}_x^{dc}\left({\rm
ZF}\right)\right]_0=\sigma_x^2\frac{\left[|h|^2\left|\sbm{f}^{H}\sbms{x}{\rm
tr}-1\right|^2+\sigma_w^2\left\|\sbm{f}\right\|^2\right]+\frac{\sigma_w^2}{\sigma_x^2}}{|h|^2\left|\sbm{f}^{H}\sbms{x}{\rm
tr}\right|^2+\sigma_w^2\left\|\sbm{f}\right\|^2}.
\end{eqnarray}
The numerator of the gradient of the above expression with respect
to $\gss{f}{}{}$ discarding the outer $\sigma_x^2$ is given by the
following expression\footnote{Necessary (hermitian) transpositions
take place, since checking the possibility of zeroing the
numerator by choosing $\g{f}$ is not affected by these
operations.}:
\begin{eqnarray}
&&\left[|h|^2|\varphi|^2+\sigma_w^2\|\g{f}\|^2\right]\left[|h|^2\left(\varphi-1\right)^{*}\gss{x}{\rm
tr}{}+\sigma_w^2\g{f}\right]\nonumber\\
&&-\left[|h|^2\varphi^{*}\gss{x}{\rm
tr}{}+\sigma_w^2\g{f}\right]\left[\frac{\sigma_w^2}{\sigma_x^2}+|h|^2\left|\varphi-1\right|^2+\sigma_w^2\|\g{f}\|^2\right],\nonumber\\
\label{eq:ddcZFnom}
\end{eqnarray}
where $\varphi=\gss{f}{}{H}\gss{x}{\rm tr}{}$. Setting
$\g{f}=\gss{f}{\rm MVU}{}$, we obtain:
\begin{equation}
-\frac{\sigma_w^2}{\sigma_x^2}\left[|h|^2\gss{x}{\rm
tr}{}+\sigma_w^2\frac{\gss{x}{\rm tr}{}}{\|\gss{x}{\rm
tr}{}\|^2}\right]\neq \g{0}. \label{eq:nomMVUE}
\end{equation}
Note that no choice of $\gss{x}{\rm tr}{}$ will zero this
expression for any $|h|^2, \sigma_w^2$. Therefore, the MVU is not
an optimal channel estimator in this case. We can state this
result more formally:
\begin{prop}\label{thrm:1}
The MVU estimator is \emph{not} an optimal channel estimator for
the task of minimizing $\left[{\rm MSE}_x^{dc}\left({\rm
ZF}\right)\right]_0$, when the channel is considered deterministic
but otherwise an unknown quantity.
\end{prop}

The question that arises in this case is how to find the optimal
channel estimator in this setup or generally how to determine a
uniformly better channel estimator for minimizing $\left[{\rm
MSE}_x^{dc}\left({\rm ZF}\right)\right]_0$. Equating
\pref{eq:ddcZFnom} to $\g{0}$ and taking the inner product of both
sides with $\g{f}$, we obtain the following necessary condition
that every optimal channel estimating filter $\g{f}$ must satisfy
given the training sequence $\gss{x}{\rm tr}{}$:
\begin{equation}
\gss{f}{}{H}\gss{x}{\rm tr}{}=\left(1+\frac{\sigma_w^2}{\sigma_x^2
|h|^2}\right).\label{eq:ncZF}
\end{equation}
A possible $\g{f}$ that satisfies this condition is
\begin{equation}
\gss{f}{\rm opt}{dc{\rm ZF}}=\left(1+\frac{\sigma_w^2}{\sigma_x^2
|h|^2}\right)\frac{\gss{x}{\rm tr}{}}{\|\gss{x}{\rm
tr}{}\|^2}=\left(1+\frac{\sigma_w^2}{\sigma_x^2
|h|^2}\right)\gss{f}{\rm MVU}{},\label{eq:optfZF}
\end{equation}
which becomes:
\begin{equation}
f_{\rm opt}^{dc{\rm ZF}}=\left(1+\frac{\sigma_w^2}{\sigma_x^2
|h|^2}\right)\frac{1}{x_{\rm tr}^{*}}\label{eq:ZF-MVUE-opt}
\end{equation}
for $B=1$. Clearly, (\ref{eq:optfZF}) is sufficient for (\ref{eq:ddcZFnom}) to become zero. However, (\ref{eq:optfZF}) has another problem, namely that
the optimal solution depends on the unknown channel $h$.

In order to deal with the dependence of the optimal estimator on
the unknown channel, we will resort to a stochastic approach. We
will assume a \emph{noninformative} prior distribution for the
unknown channel. If the real and imaginary parts of the channel
are considered bounded in the intervals $\mathcal{I}_{R}\subset
\mathbb{R}$ and $\mathcal{I}_{I}\subset \mathbb{R},$
\footnote{This assumption is usually reasonable in practice.} then
the receiver can treat them as independent random variables
uniformly distributed on $\mathcal{I}_R$ and $\mathcal{I}_I$,
respectively. The $\left[{\rm MSE}_{x}^{dc}({\rm ZF})\right]_0$ is
now replaced by $E_h^{ud}\left\{\left[{\rm MSE}_{x}^{dc}({\rm
ZF})\right]_0\right\}$, where $E_h^{ud}[\cdot]$ denotes the
expectation with respect to the joint (uniform) distribution of
the real and imaginary parts of $h$. Applying again the zeroth
order approximation and following the above analysis\footnote{Part
of this analysis is presented in
Subsection~\ref{subsec:dMSErcZF}.} we can easily show that the
eqs. \pref{eq:ncZF}, \pref{eq:optfZF} and \pref{eq:ZF-MVUE-opt}
give again the necessary condition and optimal estimators in this
case with the substitution of $|h|^2$ by $E_{h}^{ud}[|h|^2]$.

\subsubsection{Random channel}
\label{subsec:dMSErcZF}

In this case, the actual prior statistics of the channel are
known. The zeroth order symbol estimate MSE is given by
\begin{equation}
\left[{\rm MSE}_x^{ rc}({\rm ZF})\right]_0=
\sigma_x^2\frac{\left[E[|h|^2]\left|\sbm{f}^{H}\sbms{x}{\rm
tr}-1\right|^2+\sigma_w^2\left\|\sbm{f}\right\|^2\right]+\frac{\sigma_w^2}{\sigma_x^2}}{E[|h|^2]\left|\sbm{f}^{H}\sbms{x}{\rm
tr}\right|^2+\sigma_w^2\left\|\sbm{f}\right\|^2}.
\end{equation}
Differentiating this expression with respect to $\gss{f}{}{}$, we
get the numerator of the gradient which is given by
\pref{eq:ddcZFnom}\footnote{ignoring all the positive scaling
terms}, but with $|h|^2$ replaced by $E[|h|^2]$. It can be easily
shown that this numerator is different from zero if
$\g{f}=\gss{f}{\rm MVU}{}$ or $\g{f}=\gss{f}{\rm MMSE}{}$. We
therefore have a formal statement of this result:
\begin{prop}
The MVU and MMSE estimators are \emph{not} optimal channel
estimators for the task of minimizing $\left[{\rm MSE}_x^{
rc}({\rm ZF})\right]_0$, when the prior channel distribution is
known.\label{thrm:2}
\end{prop}

The optimal channel estimator $\gss{f}{\rm opt}{rc{\rm ZF}}$
satisfies \pref{eq:ncZF}, \pref{eq:optfZF} and
\pref{eq:ZF-MVUE-opt}, but with $|h|^2$ replaced by $E[|h|^2]$.


\emph{Remarks:}
\begin{enumerate}
\item Considering (\ref{eq:optfZF}) and the corresponding expression for the random channel case, we observe that
the design of the estimator with respect to the end performance metric introduces a bias to the MVU estimator in the form of scaling, leading to
a smaller value of the end performance metric than the one that we would obtain by using the MVU estimator. This bias introduction mechanism
has similarities with the introduction of bias in estimators to reduce their MSE (the MSE here is the average square distance of the parameter estimator from the true
value of the parameter)~\cite{e08}. Nevertheless, the reader may observe the conceptual differences in the motivation and goals behind the end performance metric estimator designs presented in this paper and the ideas in~\cite{e08}.
\item The claimed optimality of the derived estimators in this section but also in this paper is with respect to the zeroth order performance metrics.
These estimators turn out to be uniformly better than the MVU and MMSE estimators also when comparing against the \emph{true} end performance metrics as we demonstrate in the simulation section.
\item An alternative way to express eq. (\ref{eq:optfZF}) is
\begin{equation}
\gss{f}{\rm opt}{dc{\rm ZF}}=\left(1+\alpha\right)\gss{f}{\rm MVU}{},\label{eq:optfZF_SNR}
\end{equation}
where $\alpha=\sigma_w^2/(\sigma_x^2
|h|^2)$ is the inverse SNR at the recceiver side. Depending on how we implement the last estimator in practice, $\alpha$ turns to a tuning parameter controlling the
introduction of bias in the MVU estimator. We numerically demonstrate this very interesting aspect of the derived estimators in Figs. \ref{fig:ZFMSE_biased} and \ref{fig:ZFexMSE_biased}.
\end{enumerate}
\subsection{Discussion on the Optimal Training}

Since the channel estimator is selected in order to optimize the
final performance metric of the communication system, one may
consider the problem of selecting optimally the training vector
$\gss{x}{\rm tr}{}$ under a training energy constraint $\|\gss{x}{tr}{}\|^2\leq \mathcal{E}$ to serve the same purpose. To optimize the
training vector, one should first fix the channel estimator. This
is a ``complementary'' problem with respect to the approach that we have
followed so far. Suppose that we use either the MVU or the MMSE
channel estimators. One can observe that for $B=1$ the problem of
selecting optimally the training vector is meaningless. Therefore,
we will end up using an inferior channel estimator (i.e., the MVU or the MMSE) than the one
given by \pref{eq:ZF-MVUE-opt} and its random channel counterpart.
In the case that $B>1$, fixing for example $\g{f}=\gss{f}{\rm
MVU}{}$ one can observe that again the problem of selecting
optimally the training vector is meaningless. Consider for example
the case of $\left[{\rm MSE}_x^{rc}\left({\rm
ZF}\right)\right]_0$. We then have:
\[
\left[{\rm MSE}_x^{rc}\left({\rm
ZF}\right)\right]_0=\frac{\frac{\sigma_x^2\sigma_w^2}{\|\gss{x}{tr}{}\|^2}+\sigma_w^2}{E[|h|^2]+\frac{\sigma_w^2}{\|\gss{x}{tr}{}\|^2}},
\]
which only depends on $\|\gss{x}{tr}{}\|^2$.
Furthermore, setting $\theta=\|\gss{x}{tr}{}\|^2$, it follows that
$d\left[{\rm MSE}_x^{rc}\left({\rm ZF}\right)\right]_0/d\theta<0$ at sufficiently high SNR,
i.e., $\left[{\rm MSE}_x^{rc}\left({\rm ZF}\right)\right]_0$ is
minimized when $\|\gss{x}{tr}{}\|^2=\mathcal{E}$, which is
intuitively appealing. Therefore, any $\gss{x}{tr}{}$ with energy
equal to $\mathcal{E}$ is an equally good training vector for the
MVU estimator. Thus, for the same $\gss{x}{tr}{}$, the estimator
$\g{f}=\gss{f}{\rm opt}{rc{\rm ZF}}$ will be better than the MVU.
Similar conclusions can be reached for the MMSE estimator, as
well.

\section{Minimizing the Zeroth Order Excess MSE}
\label{sec:eMSE}

We now examine the zeroth order excess MSE in the case of the ZF
equalizer.

\subsection{ZF Equalizer with a deterministic channel}
\label{subsec:eMSEdcZF}

In this case, we have:
\begin{equation}
\left[{\rm MSE}_{xe}^{dc}\left({\rm ZF}\right)\right]_0=
\frac{|h|^2\left|\sbm{f}^{H}\sbms{x}{\rm
tr}-1\right|^2+\sigma_w^2\left\|\sbm{f}\right\|^2}{|h|^2\left|\sbm{f}^{H}\sbms{x}{\rm
tr}\right|^2+\sigma_w^2\left\|\sbm{f}\right\|^2}\left(\sigma_x^2+\frac{\sigma_w^2}{|h|^2}\right)
\end{equation}
The numerator of the gradient of the above expression with respect
to\footnote{discarding the positive scalars and considering again
the corresponding (hermitian) transpositions.} $\gss{f}{}{}$ is
given by the following expression:
\begin{eqnarray}
&&\left[|h|^2|\varphi|^2+\sigma_w^2\|\g{f}\|^2\right]\left[|h|^2\left(\varphi-1\right)^{*}\gss{x}{\rm
tr}{}+\sigma_w^2\g{f}\right]\nonumber\\
&&-\left[|h|^2\varphi^{*}\gss{x}{\rm
tr}{}+\sigma_w^2\g{f}\right]\left[|h|^2\left|\varphi-1\right|^2+\sigma_w^2\|\g{f}\|^2\right]\nonumber\\
\label{eq:eZFnom}
\end{eqnarray}
Setting $\g{f}=\gss{f}{\rm MVU}{}$, one can easily check that the
above expression becomes zero. Therefore:
\begin{prop}
The MVU \emph{is} an optimal channel estimator for the task of
minimizing $\left[{\rm MSE}_{xe}^{dc}\left({\rm
ZF}\right)\right]_0$, when the channel is considered a deterministic
but otherwise unknown quantity.
\end{prop}

\emph{Remark:} Note that even if $\left[{\rm MSE}_{xe}^{
dc}\left({\rm ZF}\right)\right]_0$ depends on the unknown channel
$h$, the optimal channel estimator does not in this case.

\subsection{ZF Equalizer with a random channel}
\label{subsec:eMSErcZF}

In this case, the prior statistics of the channel are known. The
zeroth order excess MSE is given by:
\begin{eqnarray}
\left[{\rm MSE}_{xe}^{rc}(ZF)\right]_0&=&\frac{\left|\varphi-1\right|^2(E[|h|^4]\sigma_x^2+E[|h|^2]\sigma_w^2)}{E[|h|^4]|\varphi|^2+\sigma_w^2\left\|\sbm{f}\right\|^2
E[|h|^2]}\nonumber\\
&&+\frac{\sigma_w^2\left\|\sbm{f}\right\|^2(E[|h|^2]\sigma_x^2+\sigma_w^2)}{E[|h|^4]|\varphi|^2+\sigma_w^2\left\|\sbm{f}\right\|^2
E[|h|^2]}
\end{eqnarray}
Differentiating this expression w.r.t. $\gss{f}{}{}$ and setting
$\g{f}=\gss{f}{\rm MVU}{}$ we zero the gradient. Therefore:
\begin{prop}
The MVU \emph{is} an optimal channel estimator for the task of
minimizing $\left[{\rm MSE}_{xe}^{rc}(ZF)\right]_0$, when the
channel is considered random.
\end{prop}

Via tedious calculations, we can show that the MMSE channel
estimator does not zero the gradient.

\emph{Remark:} This result is \emph{counterintuitive}: it says that
    when one has knowledge of the channel statistics but uses a ZF
    equalizer, one should ignore these statistics in choosing a
    channel estimator for minimizing the zeroth order excess MSE.

\section{Minimizing the Zeroth Order Probability of Error for the ML detector}
\label{sec:minPe}

It is straightforward to see that the decision rule given by
\pref{eq:MLdec} is equivalent to:
\begin{equation}
{\rm dec}\left[x(n)\right](h)=\arg\min_{\hat{x}(n)\in
\mathcal{X}}\left|\frac{y(n)}{h}-\hat{x}(n)\right|^2
\label{eq:MLdec1}
\end{equation}
With a given channel estimate, $h$ is replaced by $\hat{h}$ in the
last expression\footnote{Notice that this does not generalize to ISI and/or MIMO channels.}.

In the case of a perfectly known channel, the division
$y(n)/h=x(n)+w(n)/h$ results in an AWGN channel with information bearing signal power $\sigma_x^2$ and
noise variance $\sigma_w^2/|h|^2$. If only an estimate of the
channel, $\hat{h}=h+\epsilon$, is available, then the division
results in $y(n)/\hat{h}=x(n)+(w(n)-\epsilon x(n))/\hat{h}$. Here,
$\epsilon$ is the channel estimation error, which is Gaussian
distributed according to our assumptions with
$E[\epsilon]=h\left(\gss{f}{}{H}\gss{x}{\rm tr}{}-1\right)$ and
variance $\sigma_{\epsilon}^2=\sigma_w^2\|\g{f}\|^2$. Also,
$E[|\epsilon|^2]=|h|^2\left|\gss{f}{}{H}\gss{x}{\rm
tr}{}-1\right|^2+\sigma_w^2\|f\|^2$.

For the case of most common constellations and an AWGN
channel, the error probability is given by~\cite{pro95}:
\begin{equation}
P_{e}\approx aQ\left(b\sqrt{\rm SNR}\right) \label{eq:PeBPSK}
\end{equation}
where $Q(x)=(1/\sqrt{2\pi})\int_{x}^{+\infty}e^{-t^2/2}dt$ and $a,b$ are positive constants
depending on the geometry of the constellation. With a channel
estimation error, the useful signal power is again $\sigma_x^2$.
The noise variable is now $w(n)'=(w(n)-\epsilon x(n))/\hat{h}$ and
therefore $E[w(n)']=0$. For the power of the noise component, we
have:
\begin{equation}
E\left[|w(n)'|^2\right]=E\left[\left|\frac{w(n)}{\hat{h}}\right|^2\right]+E\left[\left|\frac{\epsilon}{\hat{h}}\right|^2\right]\sigma_x^2\label{eq:noisePbpsk}
\end{equation}
Here, we face again the infinite moment problem.  Using again similar arguments as
in the appendix for approximating $E[X/Y]$ by $E[X]/E[Y]$ when $Y=|\hat{h}|^2$, we
define the corresponding zeroth order version of
$E\left[|w(n)'|^2\right]$:
\begin{eqnarray}
\left\{%
\begin{array}{c}
  \left[E\left[\left|\frac{w(n)}{\hat{h}}\right|^2\right]\right]_0=\frac{\sigma_w^2}{E[|\hat{h}|^2]}=\frac{\sigma_w^2}{\left[|h|^2\left|\sbm{f}^{H}\sbms{x}{\rm tr}\right|^2+\sigma_w^2\left\|\sbm{f}\right\|^2\right]} \\
  \left[E\left[\left|\frac{\epsilon}{\hat{h}}\right|^2\right]\right]_0=\frac{E[|\epsilon|^2]}{E[|\hat{h}|^2]}=\frac{|h|^2\left|\sbm{f}^{H}\sbms{x}{\rm tr}-1\right|^2+\sigma_w^2\left\|\sbm{f}\right\|^2}{|h|^2\left|\sbm{f}^{H}\sbms{x}{\rm tr}\right|^2+\sigma_w^2\left\|\sbm{f}\right\|^2}\\
\end{array}%
\right.\label{eq:usefulRatios}
\end{eqnarray}
A variation of the error probability performance metric for any of
the commonly used linear modulation schemes, named \emph{zeroth order} error
probability, will be given by the following expression:
\begin{align}
\left[P_{e}\right]_0=&
aQ\left(b\sqrt{\frac{\left[|h|^2\left|\sbm{f}^{H}\sbms{x}{\rm
tr}\right|^2+\sigma_w^2\left\|\sbm{f}\right\|^2\right]}{\sigma_w^2/(\sigma_x^2)+\left[|h|^2\left|\sbm{f}^{H}\sbms{x}{\rm
tr}-1\right|^2+\sigma_w^2\left\|\sbm{f}\right\|^2\right]}}\right)\nonumber\\
=&aQ\left(b\sqrt{\frac{\sigma_x^2}{\left[{\rm MSE}_{x}^{dc}({\rm
ZF})\right]_0}}\right) \label{eq:PeBPSK1}
\end{align}
where we have used the \emph{zeroth order} SNR approximation given
by:
\[
\left[{\rm SNR}\right]_0=\frac{\sigma_x^2}{\left[{\rm
MSE}_{x}^{dc}({\rm ZF})\right]_0}.
\]

 Clearly, (\ref{eq:PeBPSK1}) is an artificial performance metric that appears in this paper for the
sake of our arguments. It is used as a
variation of the error probability to help us
extract useful conclusions. 

Since $Q(x)$ is a strictly decreasing function, the zeroth order
probability of error for a given channel $h$ is minimized when
$\left[{\rm MSE}_{x}^{dc}({\rm ZF})\right]_0$ is minimized.
Therefore, the results of Subsection \ref{subsec:dMSEdcZF} apply:
\begin{prop}
The MVU estimator is \emph{not} an optimal channel estimator for
the task of minimizing $\left[P_{e}\right]_0$ given the true
channel, when using any of the well-known digital modulations in a
flat-fading AWGN channel.\label{thrm:6}
\end{prop}

Suppose now that we average $\left[P_{e}\right]_0$ with respect to
\emph{any} given channel distribution. We can then make the
following statement:
\begin{prop}
The MVU estimator is \emph{not} an optimal channel estimator for
the task of minimizing the average $\left[P_{e}\right]_0$, when
using any of the well-known digital modulations in a flat-fading
AWGN channel.\label{thrm:7}
\end{prop}
\begin{proof}
Assume that the pdf of the fading coefficient magnitude is
$p(|h|)$ and $|h|\in [\alpha,\beta],\alpha,\beta\geq 0$, $\beta$
possibly equal to $+\infty$. The average $\left[P_{e}\right]_0$ is
given by the expression:
\begin{equation}
\overline{\left[P_{e}\right]}_0= \int_{\alpha}^{\beta}
aQ\left(b\sqrt{\left[{\rm SNR}\right]_0}\right)p(|h|)d|h|
\end{equation}
Assuming that the differentiation and integral operators can be
interchanged, we can set the gradient of the above expression with
respect to $\gss{f}{}{H}$ to zero to get the equation:
\begin{eqnarray*}
&&\nabla_{\sbm{f}^{H}}\overline{\left[P_{e}\right]}_0=
\int_{\alpha}^{\beta}
a \nabla_{\sbm{f}^{H}}Q\left(b\sqrt{\left[{\rm SNR}\right]_0}\right)p(|h|)d|h|=\gss{0}{}{H}\\
&&{\rm or} \\
&& \int_{\alpha}^{\beta} a\nabla_{x}Q\left(x\right)\left|_{x=b\sqrt{\left[{\rm
SNR}\right]_0}}\right.\frac{b\nabla_{\sbm{f}^{H}}{\left[{\rm
SNR}\right]_0}}{2b\sqrt{\left[{\rm
SNR}\right]_0}}p(|h|)d|h|\\
&&=\gss{0}{}{H}
\end{eqnarray*}
where in the second equation we have used the chain rule of
differentiation. $Q(x)$ is strictly decreasing in $x$, thus
\[
a\nabla_{x}Q\left(x\right)\left|_{x=b\sqrt{\left[{\rm
SNR}\right]_0}}\right.<0
\]
for any value of $h$. Also $p(|h|)\geq 0$ for every value of $h$
since it is a distribution function. Additionally, ${\left[{\rm
SNR}\right]_0}\geq 0$ for every value of $h$. Finally, the
numerator of $\nabla_{\sbm{f}^{H}}{\left[{\rm SNR}\right]_0}$ for
the MVU estimator is given by \pref{eq:nomMVUE} multiplied by $-1$
and by a positive scalar\footnote{The denominator is always
positive as a squared term.}. Therefore, it is either positive or
negative with respect to $h$ in a componentwise fashion depending
on the sign of the corresponding element in $\gss{x}{\rm
tr}{}$\footnote{Some of the entries of $\gss{x}{\rm tr}{}$ may be
zero but not all of them simultaneously.}. These arguments verify
that $\nabla_{\sbm{f}^{H}}\overline{\left[P_{e}\right]}_0\neq
\g{0}$. This concludes the proof.
\end{proof}

If we assume that the prior distribution of $h$ is known, then
instead of the MVU, one could use the MMSE channel estimator.
Plugging $\gss{f}{\rm MMSE}{}$ into the negative of
\pref{eq:ddcZFnom}, one can obtain that $
\left.\nabla_{\sbm{f}^{H}}{\left[{\rm
SNR}\right]_0}\right|_{\sbm{f}=\sbms{f}{\rm MMSE}}\neq
\g{0}\label{eq:SNRnomMMSE}$.

Since in the case of the MMSE estimator, the assumption is that we
always know the prior channel fading distribution, we can make the
following statement:
\begin{prop}
The MMSE estimator is \emph{not} optimal for the task of
minimizing $\overline{\left[P_{e}\right]}_0$, when using any of
the well-known digital modulations in a flat-fading AWGN channel.
\label{thrm:8}
\end{prop}
\begin{proof}
The result follows along the same lines as in Proposition \ref{thrm:7}.
\end{proof}

The problems of determining the optimal channel estimator for the
task of minimizing $\left[P_{e}\right]_0$ for a given channel $h$
and $\overline{\left[P_{e}\right]}_0$ was already solved in
Subsections \ref{subsec:dMSEdcZF-1} and \ref{subsec:dMSErcZF},
respectively. In the case of $\overline{\left[P_{e}\right]}_0$, we
can only assess their optimality analytically, using the following
argument: We use the upper bound\footnote{The usual Chernoff bound
can also be used.} $Q(x)<(1/x) (1/\sqrt{2\pi})e^{-x^2/2}, x>0$,
which becomes tight as $x$ increases~\cite{pro95}. In our case,
$x=b\sqrt{\left[{\rm SNR}\right]_0}$, and we have already assumed
high SNR, therefore high $\left[{\rm SNR}\right]_0$, to justify
the use of the ZF equalizer. Using this bound and the relationship
between $\left[{\rm SNR}\right]_0$ and $\left[{\rm
MSE}_x^{dc}({\rm ZF})\right]_0$, we get:
\begin{align}
Q(b\sqrt{\left[{\rm SNR}\right]_0})&<\frac{\sqrt{\left[{\rm
MSE}_x^{dc}({\rm
ZF})\right]_0}}{b\sigma_x}\frac{1}{\sqrt{2\pi}}e^{-\frac{b^2\sigma_x^2}{2\left[{\rm
MSE}_x^{dc}({\rm ZF})\right]_0}}\nonumber\\&<\frac{\sqrt{\left[{\rm
MSE}_x^{dc}({\rm
ZF})\right]_0}}{b\sigma_x}\frac{1}{\sqrt{2\pi}}
\end{align}
where the last inequality holds for large SNR and therefore small $\left[{\rm MSE}_x^{dc}({\rm ZF})\right]_0$.
The right hand side function is concave with respect to
$\left[{\rm MSE}_x^{dc}({\rm ZF})\right]_0>0$, therefore averaging
over any channel distribution, we get:
\begin{eqnarray*}
E_{h}\left[Q\left(b\sqrt{\left[{\rm
SNR}\right]_0}\right)\right]&<&\frac{\sqrt{E_{h}\left[\left[{\rm
MSE}_x^{dc}({\rm ZF})\right]_0\right]}}{b\sigma_x\sqrt{2\pi}}
\end{eqnarray*}
We can use one more time the zeroth order approximation
to approximate $E_{h}\left[\left[{\rm
MSE}_x^{dc}({\rm ZF})\right]_0\right]$.  The right hand side is
minimized when this last zeroth order approximation is
minimized. Thus, the estimators derived in Subsections \ref{subsec:dMSEdcZF-1} and
\ref{subsec:dMSErcZF} are optimal for the task of minimizing
$\overline{\left[P_{e}\right]}_0$, in the sense that they minimize
an upper bound to $E_{h}\left[Q\left(b\sqrt{\left[{\rm
SNR}\right]_0}\right)\right]$.

\emph{Remark}: Although, we have shown that the MVU and MMSE
estimators are not optimal for the task of minimizing the zeroth
order probability of error, we will see in the simulation section
that their \emph{actual} probability of error performance is
almost identical with that of the optimal estimators for the
zeroth order probability of error. This is due to two facts: first,
the zeroth order probability of error is a
variation of the actual probability of error and
second, in practice the difference in the channel estimates must
be large enough to give rise to a notable difference in the probability of error.
Nevertheless, we conjecture that such a difference may be more
clear in the case of multiple input multiple output (MIMO)
systems if tight approximations of the error probability functions are used to derive the
corresponding channel estimators.

\section{Simulations}
\label{sec:sims}

\begin{figure}
\begin{center}
\begin{tabular}{c}
  \includegraphics[scale=0.5]{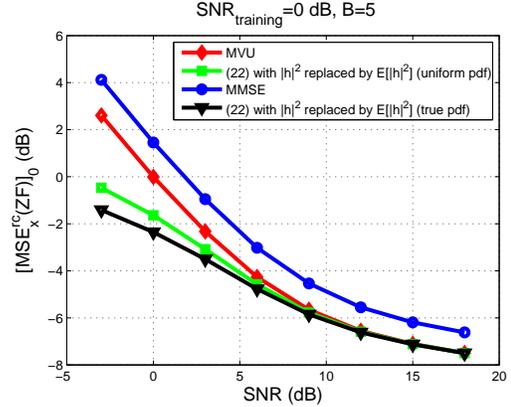}
\end{tabular}
\end{center}
\caption{$\left[{\rm MSE}_x^{rc}(\rm ZF)\right]_0$ with SNR during
training equal to $0$~dB and $B=5$. }
  \label{fig:ZFMSE_zeroth}
\end{figure}

\begin{figure}
\begin{center}
\begin{tabular}{c}
  \includegraphics[scale=0.5]{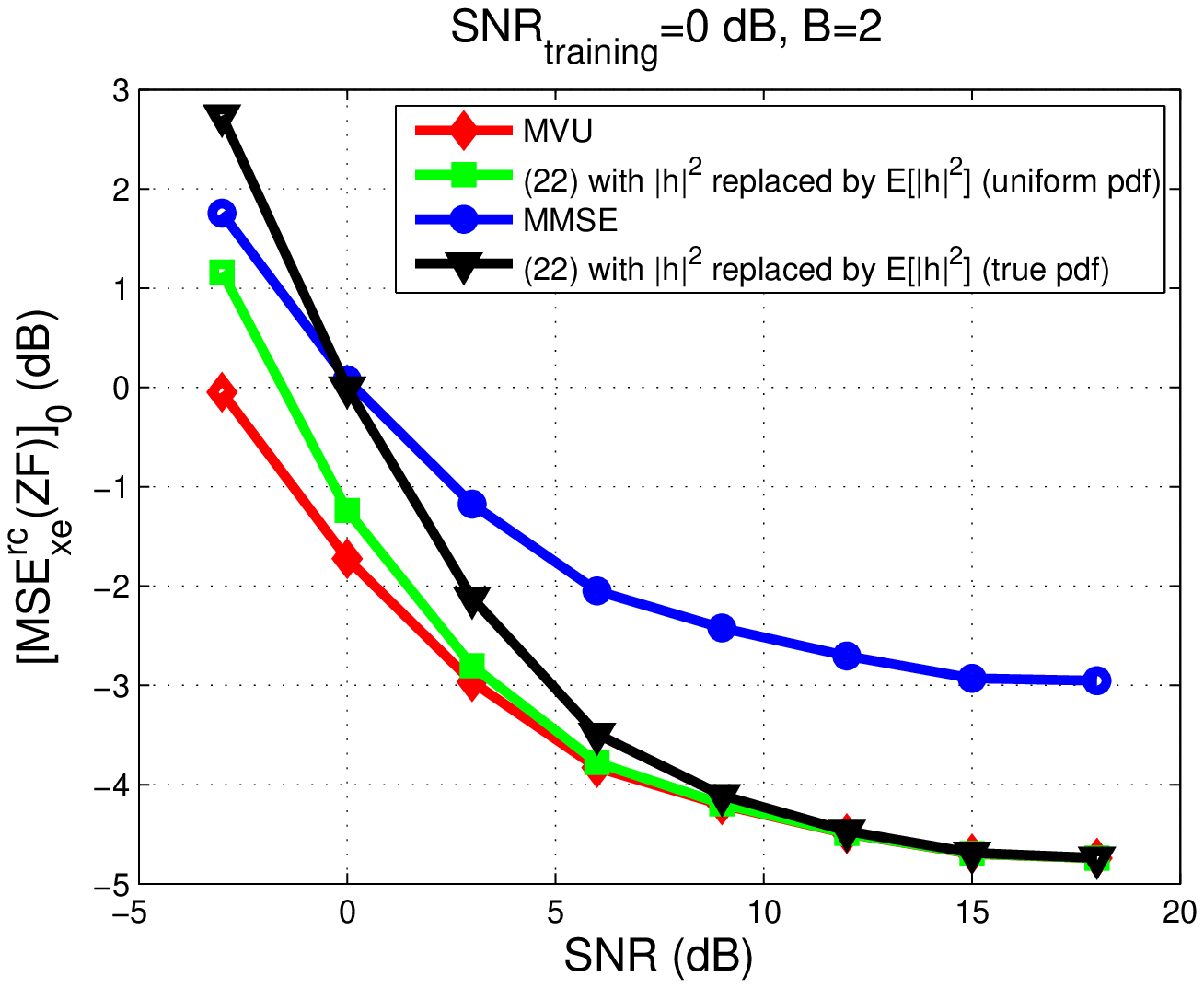}
\end{tabular}
\end{center}
\caption{$\left[{\rm MSE}_{xe}^{rc}(\rm ZF)\right]_0$ with SNR
during training equal to $0$~dB and $B=2$.}
  \label{fig:ZFexMSE_zeroth}
\end{figure}

\begin{figure}
\begin{center}
\begin{tabular}{c}
  \includegraphics[scale=0.5]{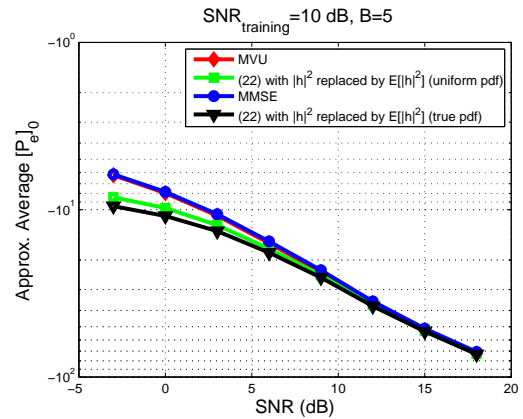}
\end{tabular}
\end{center}
\caption{Average $Q\left(\sqrt{\sigma_x^2/\left[{\rm
MSE}_{x}^{rc}(\rm ZF)\right]_0}\right)$ with SNR during training
equal to $10$~dB and $B=5$.}
  \label{fig:ZFpe_zeroth}
\end{figure}

\begin{figure}
\begin{center}
\begin{tabular}{c}
  \includegraphics[scale=0.5]{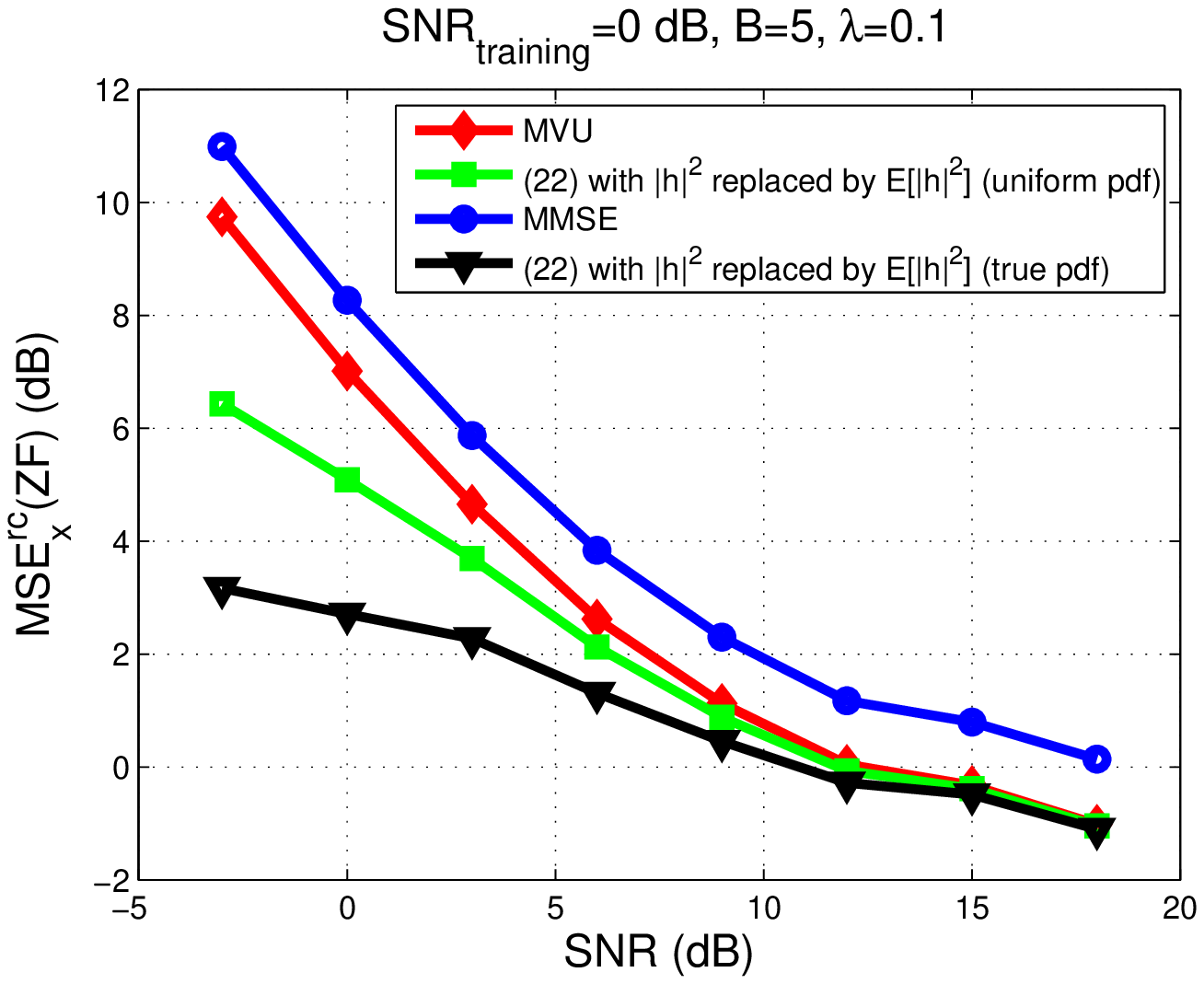}
\end{tabular}
\end{center}
\caption{${\rm MSE}_{x}^{rc}(\rm ZF)$ with SNR during training
equal to $0$~dB, $B=5$ and $\lambda=0.1$. Moreover, $E_h^{ud}[|h|^2]=3$.}
  \label{fig:ZFMSE}
\end{figure}

\begin{figure}
\begin{center}
\begin{tabular}{c}
  \includegraphics[scale=0.5]{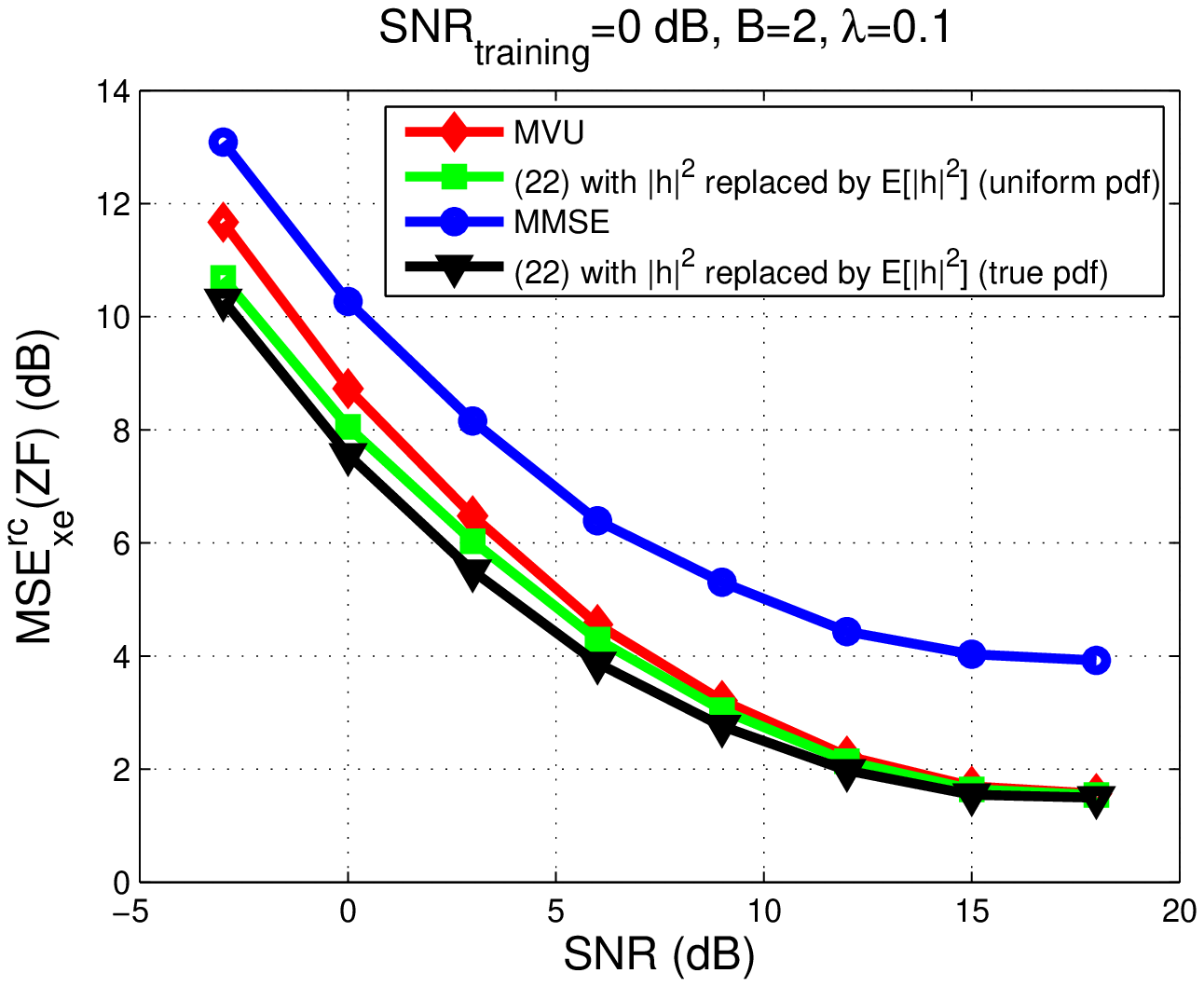}
\end{tabular}
\end{center}
\caption{${\rm MSE}_{xe}^{rc}(\rm ZF)$ with SNR during training
equal to $0$~dB, $B=2$ and $\lambda=0.1$. Moreover, $E_h^{ud}[|h|^2]=3$.}
  \label{fig:ZFexMSE}
\end{figure}

\begin{figure}
\begin{center}
\begin{tabular}{c}
  \includegraphics[scale=0.5]{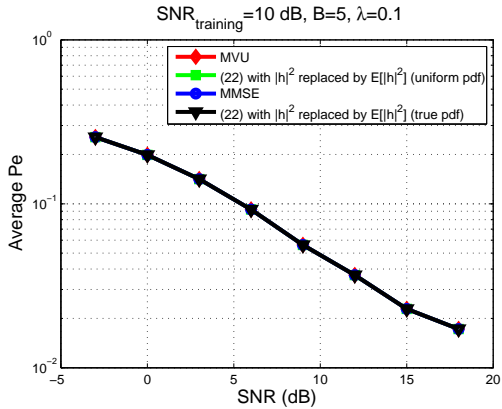}
\end{tabular}
\end{center}
\caption{Average ${P}_e$ with SNR during training equal to
$10$~dB, $B=5$ and $\lambda=0.1$. Moreover, $E_h^{ud}[|h|^2]=3$.}
  \label{fig:ZFpe}
\end{figure}

\begin{figure}
\begin{center}
\begin{tabular}{c}
  \includegraphics[scale=0.5]{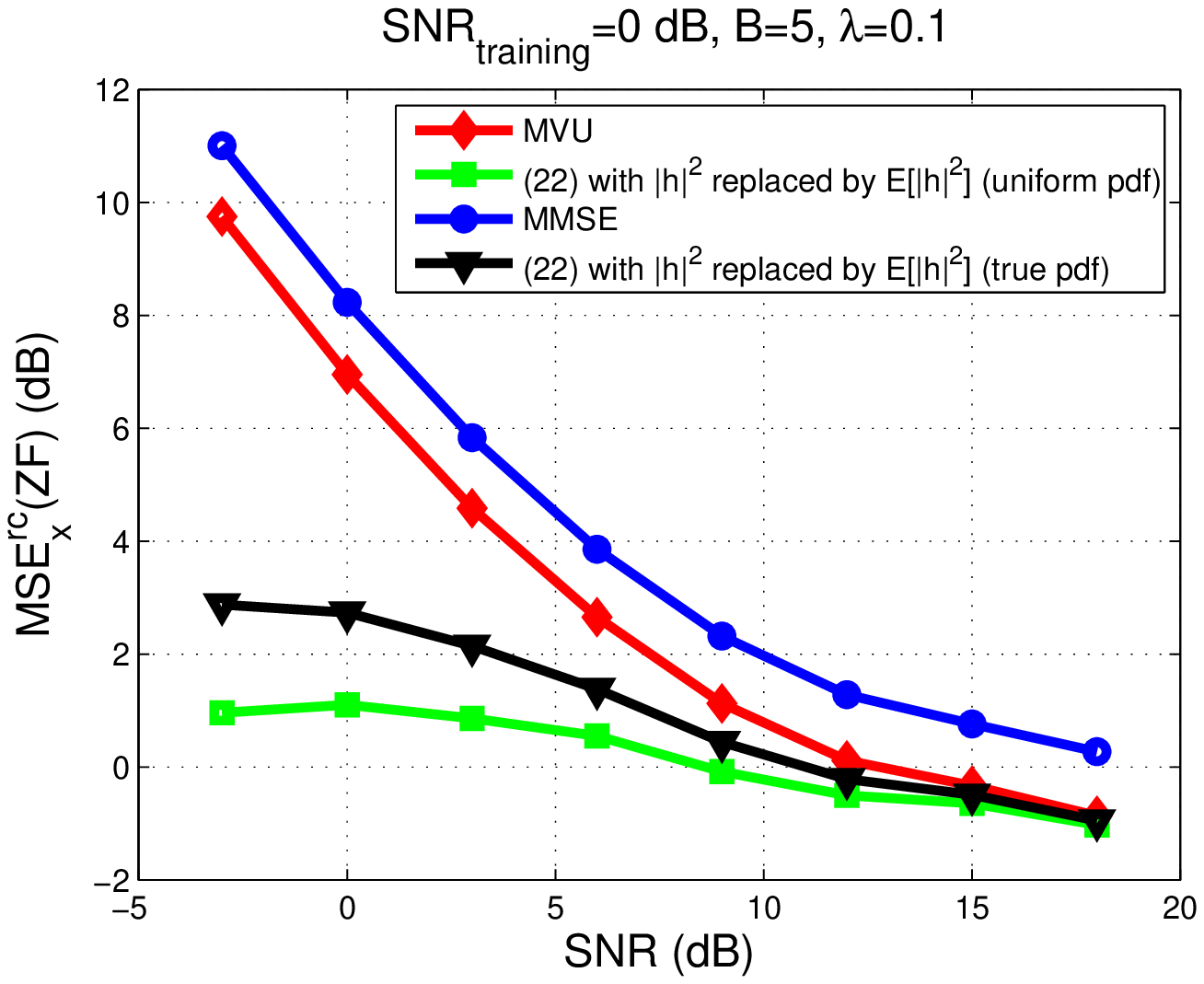}
\end{tabular}
\end{center}
\caption{${\rm MSE}_{x}^{rc}(\rm ZF)$ with SNR during training
equal to $0$~dB, $B=5$ and $\lambda=0.1$. Moreover, $E_h^{ud}[|h|^2]=1/2$.}
  \label{fig:ZFMSE_biased}
\end{figure}

\begin{figure}
\begin{center}
\begin{tabular}{c}
  \includegraphics[scale=0.5]{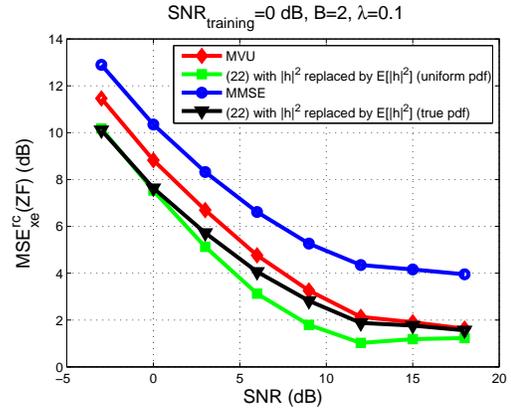}
\end{tabular}
\end{center}
\caption{${\rm MSE}_{xe}^{rc}(\rm ZF)$ with SNR during training
equal to $0$~dB, $B=2$ and $\lambda=0.1$. Moreover, $E_h^{ud}[|h|^2]=1/6$.}
  \label{fig:ZFexMSE_biased}
\end{figure}

In this section we present numerical results to verify our
analysis. In all figures, $h\sim\mathcal{CN}(0,1)$ and QPSK
modulation is assumed. The SNR during training highlights how good
the channel estimate is. The parameter $\lambda$ has
been empirically selected to be $0.1$. All schemes in
Figs.~\ref{fig:ZFMSE}-\ref{fig:ZFexMSE_biased} use
\pref{eq:WellBehEst} for the same $\lambda$. In Figs.~\ref{fig:ZFMSE}-\ref{fig:ZFpe},
$E_h^{ud}[|h|^2]$ is chosen to be $3E[|h|^2]=3$, i.e., the real and imaginary parts of
$h$ are assumed i.i.d. following a uniform distribution in $[-3/\sqrt{2},3/\sqrt{2}]$. In Figs. \ref{fig:ZFMSE_biased} and
\ref{fig:ZFexMSE_biased}, $E_h^{ud}[|h|^2]$ equals $1/2$ and $1/6$, respectively.

In Fig.~\ref{fig:ZFMSE_zeroth}, $\left[{\rm MSE}_x^{rc}(\rm
ZF)\right]_0$ is presented for $B=5$ and SNR during training equal
to $0$~dB. The derived optimal estimators in this paper are better
than the MVU and MMSE estimators. Additionally, the MVU estimator
appears to be better than the MMSE estimator for this performance
metric. This is a new observation contradicting what one would
expect and verifying the motivation of this paper.

Fig.~\ref{fig:ZFexMSE_zeroth} presents the corresponding results
for $\left[{\rm MSE}_{xe}^{rc}(\rm ZF)\right]_0$. The MVU is the
best estimator as proved. This is another example contradicting
what one would expect and verifying the motivation of this paper.

Furthermore, Fig.~\ref{fig:ZFpe_zeroth} shows the performance of
all schemes in the case of an approximation to the error
probability equal to $Q\left(\sqrt{\sigma_x^2/\left[{\rm
MSE}_{x}^{rc}(\rm ZF)\right]_0}\right)$. Here, we have assumed
that the constants $a,b$ are equal to $1$, since their specific
values are irrelevant to the purpose of this simulation plot. The
derived estimators in this paper are better than the MVU and MMSE
estimators as proved in the previous section. The difference of
the curves is present in the low SNR regime.

We now examine the performance of the derived estimators in this
paper for the true performance metrics. All the estimators are
implemented based on $\pref{eq:WellBehEst}$ to combat the infinite
moment problems.

In Fig.~\ref{fig:ZFMSE}, ${\rm MSE}_x^{rc}(\rm ZF)$ is presented
for $B=5$ and SNR during training equal to $0$~dB. The derived
optimal\footnote{The term ``optimal'' is used in this case to refer to uniformly better
estimators than the MVU/MMSE estimators and not to actually optimal estimators in the strict sense. The estimators are optimal only with respect to the zeroth order metrics.} estimators in this paper are better than the MVU and MMSE
estimators. We can see that the zeroth order approximations in
this case are satisfactory even for a low SNR during training, in
the sense that the corresponding optimal estimators outperform the
MVU and MMSE estimators for the true performance metric.
Additionally, the MVU estimator appears to be better than the MMSE
estimator for this performance metric. This is yet a new snapshot
contradicting what one would naturally expect.

Fig.~\ref{fig:ZFexMSE} presents the corresponding results for
${\rm MSE}_{xe}^{rc}(\rm ZF)$. The MVU is better than the MMSE
estimator, coinciding with the analysis based on the zeroth order
approximation. Note however that the other two estimators appear
to be better than the MVU. To obtain a well-behaved ${\rm
MSE}_{xe}^{rc}(\rm ZF)$ in this case, regularization of the same form
as in $\pref{eq:WellBehEst}$ is applied to $h$ to avoid values
around zero. In this sense, Fig.~\ref{fig:ZFexMSE} serves more as
a proof that the application-oriented estimator selection is valid
and less as an actual scenario present in the real world.

Furthermore, Fig.~\ref{fig:ZFpe} shows the performance of all schemes
in the case of the error probability performance metric. Monte Carlo
simulations have been used to compute the actual error probability. All
schemes coincide because the differences in the channel estimates
are not so large to appear in the error probability. Nevertheless,
these differences may clearly appear in a MIMO scenario if tight approximations of the error probability function are used to derive the
corresponding channel estimators.

Finally,  Figs.~\ref{fig:ZFMSE_biased} and \ref{fig:ZFexMSE_biased} demonstrate the
validity of the Remark $3$ in the end of subsection \ref{subsec:dMSEdcZF}. These plots correspond to
Figs.~\ref{fig:ZFMSE} and \ref{fig:ZFexMSE} but with $E_h^{ud}[|h|^2]=1/2$ and $1/6$, respectively. They verify that the zeroth order metrics
used in this paper are good approximations in terms of indicating the structure of uniformly better
estimators than the MVU and MMSE. Nevertheless, the zeroth order metrics cannot really determine the best
possible bias with respect to the MVU estimator that the estimators in this paper must have in order to yield
the best possible performance against the \emph{true} performance metrics. The bias terms are only optimal with respect to the
zeroth order metrics.

\section{Conclusions}
\label{sec:concl}

In this paper, application-oriented channel estimator selection
has been compared with common channel estimators such as the MVU
and MMSE estimators. We have shown that the application-oriented
selection is the right way to choose estimators in practice. We
have verified this observation based on three different
performance metrics of interest, namely, the symbol estimate MSE,
the excess symbol estimate MSE and the error probability.

\appendix

This section proposes a simplification of the ${\rm
MSE}_{x}^{dc}({\rm ZF})$ metric for the estimator given in
\pref{eq:WellBehEst} with a fixed $\lambda$. Due to the
Gaussianity of $\gss{y}{\rm tr}{}$, ${\rm MSE}_{x}^{dc}({\rm
ZF})=\infty$ for any $\g{f}\neq \g{0}$ (infinite moment problem). Using
\pref{eq:WellBehEst}, the corresponding mean square error becomes:
\begin{eqnarray}
&&\left[{\rm MSE}_{x}^{dc}({\rm ZF})\right]_{\rm reg}={\rm
Pr}\left\{|\gss{f}{}{H}\gss{y}{\rm
tr}{}|>\lambda\right\}\cdot\nonumber\\
&&E\left[\sigma_x^2\left|1-\frac{h}{\gss{f}{}{H}\gss{y}{\rm
tr}{}}\right|^2+\frac{\sigma_w^2}{|\gss{f}{}{H}\gss{y}{\rm
tr}{}|^2}; |\gss{f}{}{H}\gss{y}{\rm
tr}{}|>\lambda\right]\nonumber\\&& + {\rm
Pr}\left\{|\gss{f}{}{H}\gss{y}{\rm tr}{}|\leq \lambda\right\}
\cdot\nonumber\\
&& E\left[\frac{\sigma_x^2}{\lambda^2}\left|\lambda
\frac{\gss{f}{}{H}\gss{y}{\rm tr}{}}{|\gss{f}{}{H}\gss{y}{\rm
tr}{}|}-h\right|^2+\frac{\sigma_w^2}{\lambda^2};
|\gss{f}{}{H}\gss{y}{\rm tr}{}|\leq\lambda\right],
\end{eqnarray}
where $;$ denotes conditioning and ``reg'' signifies the use of the regularized channel estimator in \pref{eq:WellBehEst}.
To simplify this expression, we observe that
${\rm Pr}\left\{|\gss{f}{}{H}\gss{y}{\rm tr}{}|\leq
\lambda\right\}=O(\lambda^2)$, since by the mean value theorem
this probability is equal to the area of the region
$\{|\gss{f}{}{H}\gss{y}{\rm tr}{}|\leq \lambda\}$, which is of
order $O(\lambda^2)$, multiplied by some value of the probability
density function of $|\gss{f}{}{H}\gss{y}{\rm tr}{}|$ in that
region, which is of order $O(1)$. In addition,
\begin{eqnarray}
&&E\left[\frac{\sigma_x^2}{\lambda^2}\left|\lambda
\frac{\gss{f}{}{H}\gss{y}{\rm tr}{}}{|\gss{f}{}{H}\gss{y}{\rm
tr}{}|}-h\right|^2+\frac{\sigma_w^2}{\lambda^2};
|\gss{f}{}{H}\gss{y}{\rm
tr}{}|\leq\lambda\right]=\frac{\sigma_x^2}{\lambda^2}|h|^2\nonumber\\
&&+\frac{\sigma_w^2}{\lambda^2}-2\frac{\sigma_x^2}{\lambda}\Re\left\{
h^{*}\frac{\gss{f}{}{H}\gss{y}{\rm tr}{}}{|\gss{f}{}{H}\gss{y}{\rm
tr}{}|}\right\}\nonumber
\end{eqnarray}
If in addition the SNR during training is sufficiently high and the probability mass of $|\gss{f}{}{H}\gss{y}{\rm tr}{}|$ is
concentrated around $|h|$, then it
can be shown that
\begin{eqnarray}
&& E\left[\sigma_x^2\left|1-\frac{h}{\gss{f}{}{H}\gss{y}{\rm
tr}{}}\right|^2+\frac{\sigma_w^2}{|\gss{f}{}{H}\gss{y}{\rm
tr}{}|^2}|; |\gss{f}{}{H}\gss{y}{\rm
tr}{}|>\lambda\right]\approx \nonumber\\ &&
\frac{\sigma_x^2E[|\gss{f}{}{H}\gss{y}{\rm
tr}{}-h|^2; |\gss{f}{}{H}\gss{y}{\rm
tr}{}|>\lambda]+\sigma_w^2}{E[|\gss{f}{}{H}\gss{y}{\rm
tr}{}|^2; |\gss{f}{}{H}\gss{y}{\rm
tr}{}|>\lambda]}.
\end{eqnarray}
The same holds even if $\gss{f}{}{H}\gss{y}{\rm tr}{}$ is a biased estimator of $h$ at high training SNR and $|\gss{f}{}{H}\gss{y}{\rm tr}{}|$ tends to concentrate
around a value $\alpha$ bounded away from $|h|$ (and of course from $0$).

To show the last claim, we set $X=|\gss{f}{}{H}\gss{y}{\rm tr}{}-h|^2$ and $Y=|\gss{f}{}{H}\gss{y}{\rm tr}{}|^2$. Since $Y>\lambda^2$, it also holds that $E\left[Y\right]>\lambda^2$. Furthermore, it can be seen that
\begin{align}
\left|E\left[\frac{X}{Y}\right]-\frac{E[X]}{E[Y]}\right|\leq \frac{1}{\lambda^4}E\left[\left|XE[Y]-YE[X]\right|\right].\label{eq:Stochineq}
\end{align}
At high training SNR, $X\rightarrow E[X]$ and $Y\rightarrow E[Y]$ in the mean square sense and therefore it can be easily shown that the right hand side of (\ref{eq:Stochineq}) converges to $0$. To see this, notice that the Cauchy-Schwarz inequality yields
\begin{align}\label{eq:Cauchyineq1}
&\frac{1}{\lambda^4}E\left[\left|XE[Y]-YE[X]\right|\right]\leq \frac{1}{\lambda^4}\left(E\left[\left|XE[Y]-YE[X]\right|^2\right]\right)^{1/2}\nonumber\\
&=\frac{1}{\lambda^4}\left(E^2[Y]E[X^2]+E[Y^2]E^2[X]-2E[XY]E[X]E[Y]\right)^{1/2}.
\end{align}
Since $X\rightarrow E[X]$ and $Y\rightarrow E[Y]$ in the mean square sense, $E[X^2]\rightarrow E^2[X]$, $E[Y^2]\rightarrow E^2[Y]$ and $E[XY]\rightarrow E[X]E[Y]$.
For the last case, notice that
\begin{align}\label{eq:Cauchyineq2}
\left|E[XY]-E[X]E[Y]\right|&=\left|E\left[(X-E[X])(Y-E[Y])\right]\right|\nonumber\\&\leq E\left[\left|X-E[X]\right|\left|Y-E[Y]\right|\right]\nonumber\\ & \leq \sqrt{E\left[\left|X-E[X]\right|^2\right]E\left[\left|Y-E[Y]\right|^2\right]},
\end{align}
where the last inequality follows again from the Cauchy-Schwarz inequality.
By the mean square convergence of $X$ to $E[X]$ and $Y$ to $E[Y]$ the right hand side of (\ref{eq:Cauchyineq2}) tends to $0$.
Therefore, the right hand side of (\ref{eq:Cauchyineq1}) tends to $0$.

Furthermore, under the high SNR assumption the conditional
expectations can be approximated by their unconditional ones,
since for a sufficiently small $\lambda$ their difference is due to an event of probability
$O(\lambda^2)$. Therefore,
\begin{eqnarray}
&& E\left[\sigma_x^2\left|1-\frac{h}{\gss{f}{}{H}\gss{y}{\rm
tr}{}}\right|^2+\frac{\sigma_w^2}{|\gss{f}{}{H}\gss{y}{\rm
tr}{}|^2}|; |\gss{f}{}{H}\gss{y}{\rm
tr}{}|>\lambda\right]\approx\nonumber\\ &&
\left\{\frac{\sigma_x^2E[|\gss{f}{}{H}\gss{y}{\rm
tr}{}-h|^2]+\sigma_w^2}{E[|\gss{f}{}{H}\gss{y}{\rm
tr}{}|^2]}\right\}+O(\lambda^2).
\end{eqnarray}
Combining all the above results yields
\begin{equation}\label{eq:MSExdcApp}
{\rm MSE}_x^{dc}({\rm
ZF})\approx\left\{\frac{\sigma_x^2E[|\gss{f}{}{H}\gss{y}{\rm
tr}{}-h|^2]+\sigma_w^2}{E[|\gss{f}{}{H}\gss{y}{\rm
tr}{}|^2]}\right\}+O(1).
\end{equation}
The $O(1)$ term is not negligible but for sufficiently small $\lambda$ its dependence on $\g{f}$ is insignificant. Hence,
for a sufficiently small $\lambda$ and a sufficiently high SNR during training,
minimizing ${\rm MSE}_x^{dc}({\rm ZF})$ is equivalent to
minimizing the following approximation
\begin{equation}
\left[{\rm MSE}_x^{dc}\left({\rm
ZF}\right)\right]_0=\frac{E\left[|\hat{h}-h|^2\right]}{E\left[|\hat{h}|^2\right]}\sigma_x^2+\sigma_w^2\frac{1}{E\left[|\hat{h}|^2\right]}.\label{eq:MSEddcZF0_App}
\end{equation}



\begin{thebibliography}{99}
\bibitem{ase10} F.\ F.\ Abari, F.\ K.\ Sharifabad, O.\ Edfors, ``Low Complexity Channel Estimation for LTE in Fast Fading Environments for Implementation on Multi-Standard Platforms'', \emph{Proc.~ VTC Fall 2010}, Ottawa, Canada, Sept. 6--9, 2010.
\bibitem{agm07} J.\ G.\ Andrews, A.\ Ghosh, and R.\ Muhamed, \emph{Fundamentals of WiMAX: Understanding Broadband Wireless Networking,}
Prentice-Hall, 2007.
\bibitem{bsghh06} X.\ Bombois, G.\ Scorletti, M.\ Gevers, P.\ M.\ J.\ Van
den Hof, R.\ Hildebrand, ``Least Costly Identification Experiment
for Control," \emph{Automatica}, vol.~42, no.~10, pp.~1651--1662,
2006.
\bibitem{chcl03} S-H.\ Chen, W-H.\ He, H-S.\ Chen, Y.\ Lee, ``Mode Detection, Synchronization, and Channel Estimation for DVB-T OFDM Receiver'', \emph{Proc.~Globecom 2003}, San Francisco, USA, December 1--5, 2003.
\bibitem{c07} J.\ M.\ Cioffi, \emph{Course Readers for EE379A, EE379B, EE379C,
EE479}, available at: http://www.stanford.edu/group/cioffi/.
\bibitem{e08} Y.\ C.\ Eldar, \emph{Rethinking Biased Estimation: Improving Maximum Likelihood and the Cramer-Rao Bound}, Foundations and Trends in Signal Processing, vol. 1, no. 4, pp.~305--449, 2008.
\bibitem{flypz12} J.\ Fan, G.\ Ye Li, Q.\ Yin, B.\ Peng, X.\ Zhu, ``Joint User Pairing and Resource Allocation for LTE Uplink Transmission'', \emph{IEEE Trans. on Wireless Communications}, vol. 11, no. 8, pp. 2838--2847, Aug. 2012.
\bibitem{gak12} O.\ N.\ Gharehshiran, A.\ Attar, V.\ Krishnamurthy, ``Collaborative Sub-Channel Allocation in Cognitive LTE Femto-Cells: A Cooperative Game-Theoretic Approach'', \emph{IEEE Trans. on Communications}, accepted for publication, doi: 10.1109/TCOMM.2012.100312.110480.
\bibitem{hb12} M.\ K.\ Hati, T.\ K.\ Bhattacharyya, ``Digital Video Broadcast Services to Handheld Devices and a Simplified DVB-H Receiver Subsystem'', \emph{Proc.~NCC 2012}, Indian Institute of Technology Kharagpur, India, February 3--5, 2012.
\bibitem{h09} H.\ Hjalmarsson, ``System Identification of Complex and Structured
Systems," \emph{Plenary Address European Control
Conference/European Journal of Control}, vol.~15, no.~4,
pp.~275--310, 2009.
\bibitem{hxy11} S.\ Hu, J.\ Xie, F.\ Yang, ``Improved Pilot-Aided Channel Estimation in LTE Uplink'', \emph{Proc.~ICCP 2011}, Chengdu, China, October 21--23, 2011.
\bibitem{hu05} P.\ J.\ Huber, \emph{Robust Statistics}, John Wiley $\&$ Sons, 2005.
\bibitem{g05} A.\ Goldsmith, \emph{Wireless Communications},
Cambridge University Press, 2005.
\bibitem{jh05} H.\ Jansson, H.\ Hjalmarsson, ``Input Design via LMIs
Admitting Frequency-Wise Model Specifications in Confidence
Regions," \emph{IEEE Trans. Autom. Control}, vol.~50, no.~ 10,
pp.~1534--1549, 2005.
\bibitem{krhb12} D.\ Katselis, C.\ R.\ Rojas, H.\ Hjalmarsson, M.\
Bengtsson, ``Application-Oriented Finite Sample Experiment Design:
A Semidefinite Relaxation Approach," \emph{Proc.~SYSID-2012,}
Brussels, Belgium, July 2012.
\bibitem{k93} S.\ M.\ Kay, \emph{Fundamentals of Statistical Signal Processing, Volume I: Estimation
Theory}, Prentice Hall, 1993.
\bibitem{k98} S.\ M.\ Kay, \emph{Fundamentals of Statistical Signal Processing, Volume II:
Detection Theory}, Prentice Hall, 1998.
\bibitem{lcc12} D.\ Lee, M.\ Choi, S.\ Choi, ``Channel Estimation and Interference Cancellation of Feedback Interference for DOCR in DVB-T System'', \emph{IEEE Trans. on Broadcasting}, vol. 58, no. 1, pp. 87--97, March 2012.
\bibitem{lhp09} J.\ C.\ Lee, D.\ S.\ Han, S.\ Park, ``Channel Estimation based on Path Separation for DVB-T in Long Delay Situations'', \emph{IEEE Trans. on Consumer Electronics}, vol. 55, no. 2, pp. 316--321, May 2009.
\bibitem{ls12} Y.\ Liu, S.\ Sezginer, ``Iterative Compensated MMSE Channel Estimation in LTE Systems'', \emph{Proc.~ICC 2012}, Ottawa, Canada, June 10--15, 2012.
\bibitem{pcz12} D.\ L\'{o}pez-P\'{e}rez, X.\ Chu and J.\ Zhang, ``Dynamic Downlink Frequency and Power Allocation in OFDMA Cellular Networks'', \emph{IEEE Trans. on Communications}, vol. 60, no. 10, pp. 2904--2914, Oct. 2012.
\bibitem{pro95} J.\ G.\ Proakis, \emph{Digital Communications},
3rd edition, McGraw-Hill, 1995.
\bibitem{smhl02} D.\ Schafhuber, G.\ Matz, F.\ Hlawatsch, P.\ Loubaton, ``MMSE Estimation of Time-Varying Channels for DVB-T Systems with Strong Co-Channel Interference'', \emph{Proc.~EUSIPCO 2002}, Toulouse, France, Sept. 2002.
\bibitem{s11} P.\ Siebert, ``DVB: Developing Global Television Standards for Today and Tomorrow'', \emph{Proc.~ITU WT 2011}, Geneva, Switzerland, 24--27 October, 2011.
\bibitem{sacbf10} I.\ L.\ J.\ da Silva, A.\ L.\ F.\ de Almeida, F.\ R.\ P.\ Cavalcanti, R.\ Baldemair, S.\ Falahati, ``Improved Data-Aided Channel Estimation in LTE PUCCH Using a Tensor Modeling Approach'', \emph{Proc.~ICC 2010}, Cape Town, South Africa, May 23--27, 2010.
\bibitem{sy12} I.\ Siomina, D.\ Yuan, ``Analysis of Cell Load Coupling for LTE Network Planning and Optimization'', \emph{IEEE Trans. on Wireless Communications}, vol. 11, no. 6, pp. 2287--2297, June 2012.
\bibitem{yryq12} L.\ Yang, G.\ Ren, B.\ Yang, Z.\ Qiu, ``Fast Time-Varying Channel Estimation Technique for LTE Uplink in HST Environment'', \emph{IEEE Trans. on Vehicular Technology}, vol. 61, no. 9, pp. 4009--4019, Nov. 2012.
\end{thebibliography}
\end{document}